\documentclass[a4paper,UKenglish,cleveref,autoref]{lipics-v2019}

\usepackage{thm-restate}
\usepackage{stmaryrd} 

\title{Constructive expressive power of population protocols}

\author{Mikhail Raskin}{Technical University of Munich, Germany}{raskin@mccme.ru, raskin@in.tum.de}{ https://orcid.org/0000-0002-6660-5673 }{}
\authorrunning{M. Raskin}

\Copyright{Mikhail Raskin}



\funding{The project has received funding from the European Research Council (ERC) under the European Union’s Horizon 2020 research and innovation programme under grant agreement No 787367 (PaVeS)}

\acknowledgements{I would like to thank Javier Esparza, Jérôme Leroux and Chana Weil-Kennedy for useful discussions.
}

\EventEditors{John Q. Open and Joan R. Access}
\EventNoEds{2}
\EventLongTitle{42nd Conference on Very Important Topics (CVIT 2016)}
\EventShortTitle{}
\EventAcronym{}
\EventYear{2016}
\EventDate{December 24--27, 2016}
\EventLocation{Little Whinging, United Kingdom}
\EventLogo{}
\SeriesVolume{42}
\ArticleNo{23}

\hideLIPIcs
\nolinenumbers

\begin{document}
\maketitle

\begin{abstract}
Population protocols are a model of distributed computation
intended for the study of networks of independent computing agents
with dynamic communication structure.
Each agent has a finite number of states, and communication opportunities
occur nondeterministically, allowing the agents involved to change their 
states based on each other's states.
Population protocols are often studied in terms
of reaching a consensus on whether the input configuration
satisfied some predicate.

In the present paper we propose an alternative point of view.
Instead of studying the properties of inputs that a protocol can recognise,
we study the properties of outputs that a protocol eventually ensures.
We define constructive expressive power.
We show that for general population protocols and immediate observation 
population protocols
the constructive expressive power coincides with the 
normal expressive power.

Immediate observation protocols also preserve 
their relatively low verification complexity in the
constructive expressive power setting.
\end{abstract}

\section{Introduction}

Population protocols have been
introduced in \cite{conf/podc/AngluinADFP04,journals/dc/AngluinADFP06}
as a restricted yet useful subclass of general distributed protocols.
Each agent in a population protocol has a fixed amount of local storage,
and an execution consists of 
selecting pairs of agents and letting them update their states
based on an interaction.
The choice of pairs is assumed to be performed by an adversary 
subject to a fairness condition.
The fairness condition ensures
that the adversary must allow the protocol to progress.

Typically, population protocols are studied from the point of view
of recognising some properties of an input configuration.
In this context population protocols and their subclasses have been studied 
from the point of view of expressive power \cite{journals/dc/AngluinAER07},
verification complexity
\cite{conf/podc/BlondinEJM17,esparza2019parameterized,journals/corr/abs-1912-06578},
time to convergence \cite{conf/wdag/AngluinAE06,conf/wdag/DotyS15},
necessary state count \cite{blondin2019succinct},
etc.

The original target application
of population protocols and related models
is modelling networks of restricted sensors,
starting from
the original paper \cite{conf/podc/AngluinADFP04} on population protocols.
Of course, in the modern applications the cheapest microcontrollers
typically have tens of thousands of bits of volatile memory
permitting the use of simpler and faster algorithms for recognising 
properties of an input configuration.
So on the one hand, the original motivation for the restrictions in the 
population protocol model seems to have less relevance.
On the other hand, verifying distributed systems
benefits from access to a variety of restricted models
with a wide range of trade-offs between
the expressive power and verification complexity,
as most problems are undecidable
in the unrestricted case.
Complex, unrestricted, and impossible to verify distributed deployments
lead to undesirable and hard to predict and sometimes even diagnose situations 
such as so called gray failures \cite{Huang_2017} and similar.
On the other hand,
desirable behaviours of distributed systems
go beyond
consensus about a single property of initial configuration.
We find it natural to study what classes of properties 
a distributed system can eventually reach and then maintain indefinitely.

In the present paper we introduce a notion of constructive expressive
power of a protocol model,
formalising this question.
We show that for population protocols, as well as for 
a useful subclass of population protocols,
immediate observation (or one-way) population protocols,
the constructive expressive power coincides
with the classical expressive power.

We also show that immediate observation population protocols
preserve their relatively low verification complexity in the 
constructive expressive power setting.

The rest of the present paper is organised as follows.
First we provide the basic definitions and the previously
known results about expressive power from the point of view
of computing predicates.
We continue by defining constructive expressive power.
In the next section we establish the constructive expressive
power of population protocols.
Then in the following section we establish
the constructive expressive power of immediate observation protocols
and show that verification remains in $\mathbf{PSPACE}$ 
in this setting just like in the setting of computing predicates.
The paper ends with conclusion and future directions.

\section{Basic definitions}

First we define the population protocols,
as well as their subclass, immediate observation population protocols.

\begin{definition}
        A
        \emph{population protocol}
        is defined by a finite set of 
        \emph{states}~$Q$
        and a \emph{step relation}~$Step\subset{}Q^2\times{}Q^2$.
        When there is no ambiguity about the protocol,
        we abbreviate $((q_1,q_2),(q'_1,q'_2))\in{}Step$
        as $(q_1,q_2)\mapsto(q'_1,q'_2)$
        and call the quadruple
        $(q_1,q_2)\mapsto(q'_1,q'_2)$ a \emph{transition}.

        A population protocol is
        an~\emph{immediate observation} population protocol
        if  $(q_1,q_2)\mapsto(q'_1,q'_2)$
        implies $q_2=q'_2$.
        We say that an agent in the state~$q_1$
        changes its state to $q'_1$
        by \emph{observing} $q_2$
        and denote it
        $q_1\xrightarrow{q_2}q'_1$.

        A \emph{configuration}
        of a population protocol is a multiset
        of states $C:Q\to\mathbb{N}$.
        We use the notation $\Lbag{}q_1,\ldots,q_k\Rbag$
        for a multiset $C$ such that $C(q)$ is the number of 
        times $q$ occurs among $q_1,\ldots,q_k$.
	In some cases it is also convenient to interpret
	a multiset as a tuple or a vector 
	with nonnegative integer coordinates
        and denote it e.g. $C\in\mathbb{N}^Q$.
        Arithmetic operations and predicates apply to multisets
        pointwise (coordinatewise).
        The \emph{size} of a configuration~$C$
        is the sum of images of all states,
        $|C| = \sum_{q}C(q)$.
        An \emph{execution} of a population protocol
        is a finite or infinite sequence~$(C_0, C_1. \ldots)$
        of configurations such that
        for each $j$ between $1$ and the execution length we have
        $C_j\geq\Lbag{}q_1,q_2\Rbag$
        and
        $C_{j+1} = C_j - \Lbag{}q_1,q_2\Rbag + \Lbag{}q'_1,q'_2\Rbag$
        where $((q_1,q_2),(q'_1,q'_2))\in{}Step$.
        In other words, we let two agents with states
        $q_1$ and $q_2$ interact in some way permitted by the $Step$ relation.
\end{definition}

\begin{example}
        \label{ex:simple-protocols}
        Consider the set of states $\{q_0,q_1,q_2,q_3\}$.
        The step relation described by
        $(q_1,q_1)\mapsto(q_0,q_2)$,
        $(q_2,q_1)\mapsto(q_0,q_3)$,
        $(q_2,q_2)\mapsto(q_1,q_3)$,
        $(q_0,q_3)\mapsto(q_3,q_3)$,
        $(q_1,q_3)\mapsto(q_3,q_3)$,
        $(q_2,q_3)\mapsto(q_3,q_3)$
        is a population protocol but not an immediate observation population protocol.
        An example execution is:
        $\Lbag{}q_1,q_1,q_1{}\Rbag = (0.3,0,0),
	(1,1,1,0),
(2,0,0,1),
(1,0,0,2),\allowbreak
(0,0,0,3)$.
        Here we use the first two steps, then twice the fourth step.

        On the other hand, 
        $(q_1,q_1)\mapsto(q_2,q_1)$,
        $(q_2,q_2)\mapsto(q_3,q_2)$,
        $(q_0,q_3)\mapsto(q_3,q_3)$,
        $(q_1,q_3)\mapsto(q_3,q_3)$,
        $(q_2,q_3)\mapsto(q_3,q_3)$
        is an immediate observation population protocol that can also be described
        as 
        $q_1\xrightarrow{q_1}q_2$,
        $q_2\xrightarrow{q_2}q_3$,
        $q_0\xrightarrow{q_3}q_3$,
        $q_1\xrightarrow{q_3}q_3$,
        $q_2\xrightarrow{q_3}q_3$.
        An example execution is:
        $(0,3,0,0),(0,2,1,0),(0,1,2,0),(0,1,1,1),(0,0,1,2),(0,0,0,3)$.
        Here we use the first two possible steps, then the last two possible steps.
\end{example}

\begin{remark}
Note that all the configurations in an execution have the same size.
\end{remark}

We often consider executions with the steps chosen by an adversary.
However, we need to restrict adversary to ensure 
that some useful computation remains possible.
To prevent the adversary from e.g. only letting one pair of agents to interact,
we require the executions to be fair.
The fairness condition can also be described 
by comparison with random choice of steps to perform:
fairness is a combinatorial way to exclude a zero-probability set of bad executions.

\begin{definition}
        Consider a population protocol $(Q,Step)$.

        A configuration~$C'$ is \emph{reachable} from configuration~$C$ 
        if{}f there is a finite execution
        with the initial configuration $C$
        and the final configuration $C'$.

        A finite execution is \emph{fair}
        if it is not a beginning of any longer execution.

        An infinite execution $C_0, C_1, \ldots$ is \emph{fair}
        if for every configuration~$C'$
        either~$C$ is not reachable from some~$C_j$
        (and all the following configurations),
        or $C$ occurs among $C_j$ infinitely many times.
\end{definition}

\begin{example}
        The finite executions in the example~\ref{ex:simple-protocols} are fair.
\end{example}

The most popular notion of expressive power for 
population protocols is computing predicates,
defined in the following way.

\begin{definition}
        Consider a population protocol $(Q,Step)$
        with additionally defined
        nonempty set of \emph{input states}~$I\subset{}Q$,
        output alphabet $O$
        and \emph{output function} $o: Q\to{}O$.

        \emph{Support} of a configuration $C$ is the set
        of all states with nonzero images,
        The states belonging to the support of a configuration
        are also called \emph{inhabited} in the configuration.
        $\mathrm{supp}\, C = Q\setminus{}C^{-1}(0)$.

        A configuration $C$ is an \emph{input configuration}
        if its support is a subset of the set of the input states,
        $\mathrm{supp}\, C\subset I$.

        A configuration $C$ is a $b$-\emph{consensus}
        for some $b\in{}O$
        if the output function yields $b$
        for all the inhabited states,
        A configuration is a \emph{stable $b$-consensus}
        if it is a $b$-consensus together 
        with all the configurations reachable from it.
        A configuration is called 
        just a \emph{consensus}
        or a \emph{stable consensus}
        if it is a $b$-consensus (respectively stable $b$-consensus)
        for some $b$.

        A protocol \emph{computes} a function
        $\varphi:\mathbb{N}^I\to{}O$
        if{}f 
        for each input configuration~$C$
        every fair execution with initial configuration $C$
        contains a stable $\varphi(C)$-consensus.
        We usually use the protocols computing predicates,
        which corresponds to $O=\{true,false\}$.
\end{definition}

\begin{example}
        If we define the set of input states $I=\{q_1\}$,
        output set $O=\{true,false\}$
        and the output function $o(q){=}(q=q_3)$,
        both protocols in the example~\ref{ex:simple-protocols}
        compute the predicate $\varphi(C)=(C(q_1)\geq{}3)$.
\end{example}

The expressive power of population protocols
and immediate observation population protocols has been studied 
in~\cite{journals/dc/AngluinAER07}.

\begin{definition}
        A \emph{cube} is the set of configurations
        defined by a lower and an upper bound for the number of agents in each state.
        The lower bound can be zero, the upper bound can be infinite.
        A \emph{counting set} is a finite union of cubes.

        An \emph{integer cone} is the set of multisets defined
        by a base multiset (or just base) $B$
        and
        a finite number of period multisets (periods) $v_j$.
        A multiset belongs to the cone if it can be represented 
        as a sum of the base multiset and
        a non-negative integer combination of periods.
        A \emph{semilinear} set is a finite union of integer cones.
\end{definition}

\begin{theorem}[\cite{journals/dc/AngluinAER07}]
        \label{thm:rec-power}
Population protocols can compute membership in semilinear sets
and no other predicates.

Immediate observation population protocols can compute membership in counting
sets and no other predicates.
\end{theorem}

We now define constructive expressive power of protocols.

\begin{definition}
        A configuration $C$ \emph{satisfies an output condition}
        $\psi:\mathbb{N}^O\to\{true,false\}$
        if{}f $\psi(x\mapsto\sum_{o(q)=x}C(q))$ is $true$.
        In other words, we consider the multiset of the
outputs correspoding to the states of individual agents
then apply $\psi$ to this multiset.
        A configuration~$C$ \emph{ensures} an output condition $\psi$
        (given the protocol $P$)
        if every configuration $C'$ reachable from $C$ (including $C$)
        satisfies $\psi$.
        A protocol \emph{ensures} $\psi$ from configuration $C$,
        if every fair execution of $P$ starting from $C$
        reaches a configuration $C'$ that ensures $\psi$.
        A protocol \emph{ensures} $\psi$,
        if it ensures $\psi$ from every input configuration.

        An output condition $\psi$
        and size $n$
        are \emph{compatible} 
        if
        there exists a multiset $D\in\mathbb{N}^O$
of size $|D|=n$ satisfying $\psi$, i.e. $\psi(D)$ holds.
        An output condition $\psi$ 
        is \emph{size-flexible}
        if it is compatible with every (nonnegative integer) size.

        We interpret each multiset $\mathcal{D}\subset\mathbb{N}^O$
        as an output condition $\psi: D\mapsto{}D\in\mathcal{D}$.
\end{definition}

\begin{example}
        Both protocols from the example~\ref{ex:simple-protocols}
        ensure the condition $D\mapsto{}D(true)=0\vee{}D(false)=0$.
        This condition is ensured by any protocol computing a predicate.
\end{example}

\begin{remark}
        Only size-flexible conditions can be ensured.
\end{remark}

Note that the same a protocol that ensures a predicate $\psi$
also ensures every predicate $\psi'$ such that $\psi\Rightarrow\psi'$.
Therefore defining constructive expressive power of a class of protocols
requires an extra step.

\begin{definition}
        A class~$\mathcal{P}$ of population protocols
        \emph{ensures} a class $\Psi$ of output conditions
        if for each size-flexible $\psi\in\Psi$
        there is a protocol $P\in\mathcal{P}$
        that ensures $\psi$.

        A class~$\mathcal{P}$ \emph{ensures at most}
        a class $\Psi$ of output conditions
        if for each $\psi'$ ensured by some protocol $P\in\mathcal{P}$
        there is a size-flexible condition $\psi\in\Psi$ implying $\psi'$,
        i.e. $\psi\Rightarrow\psi'$.
\end{definition}

\begin{example}
        The class of protocols consisting of 
        a single protocol, namely the immediate observation protocol
        from the example~\ref{ex:simple-protocols}
        with output function satisfying
        $o(q_3)=large$ and $o(q_0)=o(q_1)=o(q2)=small$, 
        ensures the class of output conditions 
        with a single condition
        $D\mapsto(
        (D(large)=0\wedge{}D(small)\leq{}2)
        \vee
        (D(large)\geq{}3\wedge{}D(small)=0)
        )$,
        and at most that class.
\end{example}

\begin{remark}
        Note that for a given class~$\mathcal{P}$
        of protocols there is more than one
        class~$\Psi$
        of output conditions
        such that $\mathcal{P}$ ensures $\Psi$
        and at most $\Psi$.
        For example, adding a condition~$\psi'$
        that follows from some $\psi\in\Psi$
        to the class $\Psi$
        yields a class $\Psi'$
        such that $\mathcal{P}$ ensures $\Psi'$
        and at most $\Psi'$.
        We tolerate this and do not require minimality of $\Psi$
        because convenient classes of conditions,
        such as semilinear sets and counting sets,
        are not minimal.
\end{remark}

\section{Constructive power of population protocols}

In this section we establish the constructive expressive power
of general population protocols.

\subsection{Constructing semilinear predicates}

We start by providing a positive result,
showing that every predicate that can be computed,
can also be ensured.

\begin{theorem}
        \label{thm:gpp-con-lower}
        The class of all population protocols ensures 
        the class of semilinear output conditions.
\end{theorem}

The theorem follows from the following lemmas:

\begin{lemma}
        \label{lm:cone-divisibility}
        Each integer cone $\mathcal{C}$ can be represented
        as a finite union of cones
        such that all the periods of each cone
        have the same size.
\end{lemma}

\begin{proof}
        Let the cone $\mathcal{C}$ have base~$B$ 
        and periods~$v_j$ for $1\leq{j}\leq{n}$.
        Let $L$ be the least common multiple of $|v_j|$.
        Consider all the bases $B+\sum_{j} r_j v_j$
        where $0\leq{}r_j<\frac{L}{|v_j|}$.
        We consider all the cones $\mathcal{C}_{r_1,\ldots,r_n}$
        with such bases
        and periods $\frac{L}{|v_j|}v_j$.
        The size of all periods is $L$;
        it is easy to see that
        $\mathcal{C}=\bigcup{}\mathcal{C}_{r_1,\ldots,r_n}$.
\end{proof}

\begin{lemma}
        \label{lm:gpp-weak-con-cone}
        For each integer
        cone~$\mathcal{C}\subset\mathbb{N}^O$
        with all periods having the same size
        there is a protocol $P$
        such that 
the protocol
        $P$ ensures membership in $\mathcal{C}$
from each input configuration~$C$
        of size compatible with $\mathcal{C}$.
\end{lemma}

\begin{proof}
        Let the cone~$\mathcal{C}$ have base~$B$ and periods~$v_j$.
        Let $B^{(1)},\ldots,B^{(|B|)}$ and $v_j^{(1)},\ldots,v_j^{(|v_j|)}$
        for each $1\leq{}j\leq{}n$
        be sequences of output values enumerating the corresponding multisets
        with correct multiplicities.

        Note that the sizes compatible with $\mathcal{C}$
        are sums of $|B|$ and a non-negative multiple of $|v_1|$
        as all the periods have the same length.
        In the special case of no periods,
        only size $|B|$ is compatible with $\mathcal{C}$.
        The idea of the construction is to select $|B|$ agents
        to output the base multiset,
        then split the remaining agents in groups of $|v_1|$ agents
        with the outputs forming the multiset $v_1$.
        The split is performed by allowing some agents
        to recruit unassigned agents to positions in groups of size $|v_1|$.
        If two such agents interact, one of them
        disbands the current group and becomes unassigned.
        Eventually we will have just one active agent having recruited
        every other agent in some group, and the last group is complete
        if{}f the total number of agents available is divisible by $|v_1|$.

        The protocol has the following states:
        $q_{B,j}$ for $1\leq{j}\leq{}|B|$,
        and 
        $q_{v,i\uparrow}$,
        $q_{v,i\downarrow}$,
        $q_{v,i\odot}$,
        for $1\leq{i}\leq{}|v_1|$.
        The state $q_{v,1\downarrow}$ 
        plays a special role in the protocol 
        and will also be denoted as $q_\bot$.
        If there are no periods,
        we formally define $q_{B,1}$ to be also called
        $q_{v,1\uparrow}$ and $q_{v,1\downarrow}=q_\bot$. 

        The only input state is $q_{B,1}$ if $|B|>0$
        and $q_{v,1\uparrow}$ otherwise.
        The transitions are as follows:
        \begin{itemize}
                \item $(q_{B,j},q_{B,j})\mapsto(q_{B,j},q_{B,j+1})$
                        for all $1\leq{}j<|B|$;
                \item $(q_{B,|B|},q_{B,|B|})\mapsto(q_{B,|B|},q_{v,1\uparrow})$;
                \item $(q_{v,i\uparrow},q_{v,j\uparrow})\mapsto
                        (q_{v,i\uparrow},q_{v,j\downarrow})$
                        for 
                        $1\leq{i}\leq{}|v_1|$
                        and 
                        $1\leq{}{j}\leq{}|v_1|$ 
\\ (for $j=1$
 we produce $q_{v,1\downarrow} = q_{\bot}$);
                \item $(
                        q_{v,i\uparrow},q_\bot
                        \mapsto
                        q_{v,(i+1)\uparrow},q_{v,i\odot}
                        )$
                        for 
                        $1\leq{i}<|v_1|$;
                \item $(
                        q_{v,|v_1|\uparrow},q_\bot
                        \mapsto
                        q_{v,1\uparrow},q_{v,|v_1|\odot}
                        )$.
                \item $(
                        q_{v,i\downarrow},q_{v,i\odot}
                        \mapsto
                        q_{v,(i-1)\downarrow},q_\bot
                        )$
                        for 
                        $2\leq{}{i}\leq{}|v_1|$
\\ (for $i=2$ the first agent also switches to the state
$q_{v,1\downarrow}=q_\bot$).
        \end{itemize}

        The output function
        yields $B^{(j)}$ for $q_{B,j}$
        and
        $v_1^{(j)}$ for 
        $q_{v,j\uparrow}$,
        $q_{v,j\downarrow}$,
        and
        $q_{v,j\odot}$.

        It is easy to see that if there are at least $|B|$ agents, 
        all $|B|$ states $q_{B,j}$ will eventually have exactly one agent,
        with the remaining agents switching to the state $q_{v,1\uparrow}$
        at some times during the execution.

        \emph{Claim.} At every moment
        all the agents not in the states $q_{B,j}$
        can be split into groups
        with one agent in each state $q_{v,j\odot}$ for $j$ from $1$ to some $k$,
        plus one agent in the state
        $q_{v,(k+1)\uparrow}$ or $q_{v,(k+1)\downarrow}$
        unless $k=|v_1|$.
        Some of such groups contain only one agent
        in the state
        $q_{v,1\uparrow}$
        or
        $q_{v,1\downarrow}=q_\bot$.

        Indeed, this is true initially as there are either no agents
except in the state $q_{B,1}$
        or each agent forms a group being in the state $q_{v,1\uparrow}$,
        and
        each transition preserves the desired property.

        In the following we will use the fairness condition
        to claim that if some property of configuration
        can always be destroyed
        by some transitions,
        it will eventually stop holding in any fair execution.

        Note that once there are only $|B|$ agents in the states
        $q_{B,j}$,
        the number of agents in the states
        $q_{v,j\uparrow}$
        cannot increase anymore,
        but will sometimes decrease
        until there is exactly one such agent.
        After that point, the sum of the indices of 
        all the agents in the states
        $q_{v,j\downarrow}$ with $j>1$
        will only decrease until it becomes zero.
        Afterwards the number of agents in the state $q_\bot$
        will only decrease until it reaches zero.

        At that moment,
        if the difference between the number of agents and $|B|$
        is divisible by $|v_1|$,
        all agents will be divided into 
        one group of $|B|$ agents with the multiset of outputs $B$
        and some groups of $|v_1|$ agents each
        having multisets of outputs equal to $v_1$.
        As all other sizes are not compatible with $\mathcal{C}$,
        this concludes the proof.
\end{proof}

\begin{lemma}
        \label{lm:semilinear-projection}
        The set of sizes compatible with an integer cone
        is a one-dimensional integer cone.
\end{lemma}

\begin{proof}
        We identify the multisets with one-element domain with natural numbers.
        The base of the cone is the size $|B|$ of the base configuration,
        and the periods are the sizes  $|v_j|$ of the periods.
        We observe that using the same coefficients for
        non-negative integer combinations
        proves that the constructed one-dimensional cone
        contains exactly the sizes compatible with the original cone.
\end{proof}

\begin{proof}[Proof of the theorem~\ref{thm:gpp-con-lower}]
        Consider a size-flexible 
        semilinear set~$S\subset{}\mathbb{N}^O$.
        We construct our protocol as a synchronous product of multiple 
        sub-protocols.

        The set $S$
        can be represented
        as a union of cones
        $\bigcup_{j=1}^{n}\mathcal{C}_j \subset{}S$
        each having periods of the same size.
        We run synchronous product of $2n$ protocols, 
        $P_j^{con}$ ensuring membership in $\mathcal{C}_j$ for all compatible sizes
        (using the lemma~\ref{lm:gpp-weak-con-cone}),
        and
        $P_j^{rec}$ computing compatibility of configuration size with $\mathcal{C}_j$
        (this predicate is semilinear by the lemma~\ref{lm:semilinear-projection}
        thus it can be computed by the theorem~\ref{thm:rec-power}).
        The global output function is
        the output corresponding to construction
        of the first cone that is expected to be compatible with configuration size,
        $o((
        q_1^{con},q_2^{con},\ldots,q_n^{con},
        q_1^{rec},q_2^{rec},\ldots,q_n^{rec},
        )) = o^{con}_j (q_j^{con})$ where $j = \min{}k: o_k^{rec}(q_k^{rec})=true$.
        If there is no such cone, we return the first element of the output set.

        Eventually, all the protocols $P_j^{rec}$ 
        will converge to a stable consensus representing
        the true value of size compatibility.
        Therefore from some time on we will just use the output of 
        $P_j^{con}$ corresponding to the first size-compatible cone,
        which will be in $S$ from some point on by the lemma~\ref{lm:gpp-weak-con-cone}.
\end{proof}

\subsection{Upper bound on constructive expressive power of population protocols}

In this section we provide a matching upper bound for constructive
expressive power.

\begin{theorem}
        \label{thm:pp-con-only-semilinear}
        The class of all population protocols ensures at most
        the class of semilinear output conditions.
\end{theorem}

The proof uses the fact from~\cite{conf/birthday/Leroux12},
describing the structure of reachability sets of VAS,
a more general model than population protocols.

\begin{definition}
        An \emph{asymptotic integer cone}
        is a set of multisets
        defined
        by a base multiset (or just base) $B$
        and
        a possibly infinite set of period multisets (periods) $v_j$.
        We require that the domain
        of multisets $B$ and $v_j$ is finite;
        let its size be $n$.
        We require that the convex hull of the origin and all the periods
        interpreted as vectors in $\mathbb{Q}^n$
        is definable in $(\mathbb{Q},+,>)$.
        A multiset belongs to the cone if it can be represented 
        as a sum of the base multiset and
        a non-negative integer combination of periods.
        An \emph{almost semilinear} set is 
        a finite union of asymptotic integer cones.
\end{definition}

\begin{definition}
        The \emph{pre-image} of a set of configurations $X$
        is the set $\mathrm{pre}^*(X)$ such that 
        $C\in\mathrm{pre}^*(X)$ if{}f
        there is some $C'\in{}X$ reachable from $C$.
        The \emph{post-image} of $X$ is the set 
        $\mathrm{post}^*(X)$ such that
        $C\in\mathrm{post}^*(X)$ 
        if{}f it is reachable from some $C'\in{}X$.
\end{definition}

\begin{theorem}[\cite{conf/birthday/Leroux12}, restriction of Corollary 6.3]
        For any semilinear sets of configurations $X$ and $Y$,
        the sets 
        $\mathrm{post}^*(X)\cap{}Y$
        and
        $X\cap{}\mathrm{pre}^*(Y)$
        are almost semilinear.
\end{theorem}

We also use the results on structure of mutual reachability 
from \cite{journals/acta/EsparzaGLM17}.

\begin{definition}
        A configuration $C$ is a \emph{bottom configuration}
        if for each configuration $C'$ reachable from $C$,
        the configuration $C$ is reachable from $C'$.
\end{definition}

\begin{theorem}[\cite{journals/acta/EsparzaGLM17},
        lemma 3 and proposition 14]

        Each fair execution of a population protocol reaches a bottom configuration.
        The set of bottom configurations is semilinear.
\end{theorem}

\begin{proof}[Proof of the theorem~\ref{thm:pp-con-only-semilinear}]
        Each fair execution reaches a reachable bottom configuration
        and then reaches it infinitely many times.
        Thus any output condition ensured by the protocol
        is satisfied by the output corresponding 
        to any bottom configuration
        that is reachable from some input configuration.
        It suffices to find a size-flexible semilinear set 
        of reachable bottom configurations~$\mathcal{B}$,
        as its image under the output function will also be
        size-flexible and semilinear.

        The proof idea is to consider under-approximations 
        of the set of reachable bottom configuration
        using finite subsets of of periods.
        Observe that compatibility with any specific size
        can be demonstrated using just a finite number of periods;
        a simple divisibility argument shows
        that covering a finite set of sizes is sufficient.

        We know that the set of bottom configuration is semilinear,
        and therefore the set of bottom configuration reachable from
        input configurations is almost semilinear.
        Let $B_1,\ldots,B_s$ be the bases of corresponding asymptotic
        integer cones.
        Fix some enumeration $(v_{i,j})$ of the periods of these cones,
        where $v_{i,j}$ is the $j$-th period of the $i$-th cone.
        Let $M=\mathrm{max}_j{}|B_j|$ be the maximum size of a base.
        Let $L$ be the least common multiple of $|v_{i,1}|$
        for all $i$ from $1$ to $s$ corresponding to the cones
        with at least one period.

        Note that compatibility with any given size 
        can be demonstrated using a finite number
        of periods.
        Let $K$ be the maximal number of a period used 
        to demonstrate compatibility with any size
        up to $M+L$.
        We show that the semilinear set $\hat{\mathcal{B}}$
        consisting of the integer cones 
        with bases $B_i$
        and periods $v_{i,j}$ for $j\leq{}K$
        is size-flexible.
        Indeed, consider any size $S>M+L$.
        Let $r$ be the remainder of $S-M-1$ modulo $L$.
        Consider the size $M+1+r>M$.
        As this size is strictly larger than all the base sizes,
        its compatibility with $\mathcal{B}$
        has to be shown using an integer cone
        with  base $B_i$ and a nonempty set of periods $\{v_{i,j}\}$.
        Moreover, it can be demonstrated using only the periods
        $v_{i,1},\ldots,v_{i,K}$,
	as $r\leq{}L-1$ and thus $M+1+r\leq{}M+L$.
        We have $C=B_i+\sum_{j=1}^{K}a_j{}v_{i,j}\in\hat{\mathcal{B}}$,
        $|C|=M+1+r$.
        As $L$ is divisible by $|v_{i,1}|$,
        we can add the period $v_{i,1}$ to the configuration $C$
        exactly $\frac{S-|C|}{|v_{i,1}|}$ times:
        $C'=C+\frac{S-|C|}{|v_{i,1}|}v_{i,1}\in\hat{\mathcal{B}}$
        and $|C'|=S$.
        This concludes the proof.
\end{proof}

\section{Constructive expressive power of immediate observation protocols}

In this section we switch to the study of constructive expressive power
for a subclass of population protocols.
namely immediate observation population protocols.

From the point of view of computing predicates, they have lower 
but still significant expressive power, 
but benefit from a much lower verification computational complexity
than the general protocols, namely $\mathbf{PSPACE}$-complete.
We show that these properties also hold in the context
of ensuring protocols.

\subsection{Constructing counting sets}

Just like in the case of general population protocols,
we start by providing the feasibility result.

\begin{theorem}
        \label{thm:io-con-lower}
        The class of immediate observation population protocols ensures 
        the class of counting output conditions.
\end{theorem}

The proof is similar to the proof of the theorem~\ref{thm:gpp-con-lower}.

\begin{lemma}
        \label{lm:io-weak-con-point}
        For every multiset $D\in\mathbb{N}^O$
        there is an immediate observation protocol
        with a single input state
        that ensures equality to $D$ from 
        the input configuration of size $|D|$.
\end{lemma}

\begin{proof}
        Let an enumeration $D^{(j)}$ for $j$ from $1$ to $|D|$ 
        contain each element $x\in{}O$ exactly $D(x)$ times.
        The protocol has the input state $q_1$
        and other states $q_2,\ldots,q_{|D|}$.
        The transitions are $q_{j}\xrightarrow{q_j}q_{j+1}$,
        and the output function is $o(q_j)=D^{(j)}$.
        It is easy to see that in a fair execution 
        all $|D|$ agents will have different states
        and therefore produce output $D$.
\end{proof}

\begin{lemma}
        \label{lm:io-weak-con-ray}
        For a multiset $D\in\mathbb{N}^O$
        and output value $x\in{}O$
        consider
        the cube with the lower bounds
        specified by $D$
        and the upper bounds specified by $D$ 
        except for infinite upper bound for the value $x$.
        Then
        there is an immediate observation protocol
        with a single input state
        that ensures membership in 
        that cube
        from
        each input configuration of size at least $|D|$.
\end{lemma}

\begin{proof}
        Again,
        let an enumeration $D^{(j)}$ for $j$ from $1$ to $|D|$ 
        contain each element $x\in{}O$ exactly $D(x)$ times.
        The protocol has the input state $q_1$
        and other states $q_2,\ldots,q_{|D|},q_{|D|+1}$.
        The transitions are $q_{j}\xrightarrow{q_j}q_{j+1}$,
        and the output function is $o(q_j)=D^{(j)}$,
        $o(q_{|D|+1})=x$.
        It is easy to see that in a fair execution 
        with at least
        $|D|$ agents,
        states $q_1,\ldots,q_{|D|}$
        will contain one agent each
        with the rest of the agents in the state $q_{|D|+1}$.
        Such a configuration will produce the output
        differing from $D$ only by increasing the multiplicity of $x$,
        as required.
\end{proof}

\begin{proof}[Proof of the theorem~\ref{thm:io-con-lower}]
        Consider a size-flexible counting constraint $\psi$.
        It has to contain a cube with at least one infinite upper bound.
        Consider the smallest multiset $D$ in that cube,
        and the output value $x$ having an infinite upper bound.
        Let $P^{con}_\infty$ be the protocol corresponding to $D$ and $x$
        by the lemma~\ref{lm:io-weak-con-ray}.
        Let $P^{con}_j$ for $0\leq{}j<|D|$ be the protocol
        constructed by the lemma~\ref{lm:io-weak-con-point}
        for some multiset of size $j$ satisfying the constraint $\psi$.

        By the theorem~\ref{thm:rec-power}
        there are immediate observation population protocols 
        $P^{rec}_j$ recognising equality of input size to $j$ respectively.

        We consider the synchronous of all these protocols
        and define the output function to be the output of $P_j^{con}$
        for the minimal $j$ such that $P_j^{rec}$ outputs $true$,
        or the output of $P_\infty^{con}$ if none does.
        Eventually, all the protocols $P_j^{rec}$
        provide correct configuration size information to each agent,
        and thus the outputs of the same $P_j^{con}$ or $P_\infty^{con}$
        protocol are used by all the agents.
        By construction, the multiset of these outputs 
        satisfies $\psi$ from some moment on.
\end{proof}

\subsection{Structure of bottom configurations of immediate observation protocols}

In this section we prove a structural lemma about the structure
of bottom configurations of immediate observation protocols,
which also implies an upper bound on constructive expressive power.

\begin{theorem}
        \label{thm:io-bottom}
        The set of bottom configurations
        of an immediate observation population protocol
        is a counting set.
\end{theorem}

We use the pruning techniques from \cite{conf/apn/EsparzaRW19}.
The pruning approach  is based on
deanonymisation of the agents, 
giving agents identities and arbitrarily picking which specific agent 
performs the observation at each step in the execution.
Note that there are usually many ways to deanonymise a single execution.

\begin{lemma}[\cite{conf/apn/EsparzaRW19}, Pruning Lemma]
        \label{lm:pruning}
        Consider an immediate observation population protocol
        with the set of states $Q$,
        and a configuration $C'$
        reachable from  another configuration $C$.
	Consider an execution $E$ from $C$ to $C'$ 
	and its deanonymisation such that more than $|Q|$
	agents go from some state $q$ to a state $q'$,
	where $q$ and $q'$ might be the same state.
	Then there is an execution from $C-\Lbag{}q\Rbag$
	to $C'-\Lbag{}q'\Rbag{}$
	and its deanonymisation $E'$
	where one less agent goes from $q$ to $q'$
	and for all other pairs of states the same 
	number of
	agents go between them.
\end{lemma}

\begin{proof}[Proof of the theorem~\ref{thm:io-bottom}]
        Consider an immediate observation population protocol
        with the set of states $Q$.

        We prove the claim in the following equivalent form.
        For each bottom configuration $B$,
        each configuration $C\geq{}B$
        such that for each $q$ in the support of $C-B$
        we have $B(q)\geq|Q|^4$
        is also a bottom configuration.

        The reformulated claim is proven by induction
        over the size $|C|$ if the configuration $C$.
        If $|C|=|B|$ we have $C=B$ and the claim 
        is obviously true.
	Otherwise let $q_0$ be such a state that
	$C(q_0)>B(q_0)\geq{}|Q|^4$.
	Consider any configuration $C'$ reachable from $C$
	and a deanonymised execution $E$ from $C'$ to $C$.
	By the pigeonhole principle, there is a state $q_1$
	such that more than $|Q|^3$ agents
	go from $q_0$ to $q_1$.

	We can prune the execution to obtain an execution
	$E^\downarrow$
	from
	$C-\Lbag{}q_0\Rbag{}$
	to
	$C'-\Lbag{}q_1\Rbag{}$.
	Note that $C-\Lbag{}q_0\Rbag{}\geq{}B$.
	By the induction hypothesis,
	$C-\Lbag{}q_0\Rbag{}$ is a bottom configuration.
	Thus there is an execution
	from
	$C'-\Lbag{}q_1\Rbag{}$
	to
	$C-\Lbag{}q_0\Rbag{}$.
	Let's deanonymise it;
	if going from 
	$C-\Lbag{}q_0\Rbag{}$
	to
	$C'-\Lbag{}q_1\Rbag{}$
	and back permutes the agents,
	repeat this procedure the order of the permutation
	times.
	This yields an execution
	$E^{\downarrow-1}$
	from
	$C'-\Lbag{}q_1\Rbag{}$
	to
	$C-\Lbag{}q_0\Rbag{}$
	such that it moves all the agents back
	to the same states, undoing the execution $E^\downarrow$ from 
	$C-\Lbag{}q_0\Rbag{}$
	to
	$C'-\Lbag{}q_1\Rbag{}$.

	Adding a non-interacting agent in the state $q_1$
	provides an execution from $C'$ to $C-\Lbag{q_0}\Rbag + \Lbag{q_1}\Rbag$
	with most of the agents going in the opposite direction
	compared to $E$.
	Combining this with the execution $E$,
	we obtain an execution from $C$ to $C-\Lbag{q_0}\Rbag + \Lbag{q_1}\Rbag$
	with only one agent going from $q_0$ to $q_1$ and the rest
	eventually going from their states back to the same states.
	Now we have at least $C(q_0)-1\geq{}|Q|^4$ agents
	going from $q_0$ to $q_0$.
	We apply the lemma~\ref{lm:pruning}
	for obtain an execution
	from
	$C-\Lbag{q_0}\Rbag$
	to
	$C-2\times\Lbag{q_0}\Rbag + \Lbag{q_1}\Rbag$.
	We use again that 
	$C-\Lbag{q_0}\Rbag$
	is a bottom configuration
	to obtain an execution from 
	$C-2\times\Lbag{q_0}\Rbag + \Lbag{q_1}\Rbag$
	to
	$C-\Lbag{q_0}\Rbag$.
	Adding a non-interacting agent in the state $q_0$
	provides an execution from 
	$C-\Lbag{q_0}\Rbag + \Lbag{q_1}\Rbag$
	to $C$,
	proving that we can reach $C$ from $C'$ 
	via 
	$C-\Lbag{q_0}\Rbag + \Lbag{q_1}\Rbag$.
	As $C'$ was an arbitrary configuration reachable from $C$,
	this concludes the proof that $C$ is a bottom configuration.
\end{proof}

This structural result implies the desired constructive expressive power
upper bound using one more lemma from~\cite{conf/apn/EsparzaRW19}.

\begin{lemma}[\cite{conf/apn/EsparzaRW19}]
        The set of configurations
        reachable from a given counting set of configurations
        is also a counting set.
\end{lemma}

\begin{theorem}
        The class of immediate observation population protocols ensures at most
        the class of counting output conditions.
\end{theorem}

\begin{proof}
        We observe that the set of reachable bottom configurations
        is a counting set as an intersection 
        of the counting set of bottom configurations
        and the counting set of configurations reachable from input configuration.
        Then its image under the output function is a size-flexible counting set
        implying the ensured output condition.
\end{proof}

\subsection{Verification complexity for constructive immediate observation protocols}

In this section we show that the relatively low verification complexity for 
immediate observation protocols is also applicable in the case of 
constructive expressive power.

\begin{theorem}
        \label{thm:io-con-verify}
        The problem of verifying whether a given immediate observation protocol $P$
        ensures a given counting output condition $\psi$ given as a list of cubes
        with bounds written in unary
        is in $\mathbf{PSPACE}$.
\end{theorem}

Here we use a convenient complexity claim from \cite{journals/corr/abs-1912-06578}.

\begin{lemma}[\cite{journals/corr/abs-1912-06578}, claim in the proof of Theorem 4.50]
        Given two functions that produce counting sets 
        with membership and emptiness in $\mathbf{PSPACE}$
        and at most exponential constants
        from a counting set and a protocol,
        their boolean combinations as well 
        as pre-image and post-image
        also have the same properties.

        Let $\mathcal{S}_1$ and $\mathcal{S}_2$ be two functions 
        that take as arguments an IO protocol $P$ and
        a counting constraint $X$,
        and return counting sets $\mathcal{S}_1(P,X)$
        and $\mathcal{S}_2(P,X)$ respectively.
        Assume that $\mathcal{S}_1(P,X)$  and $\mathcal{S}_2(P,X)$ 
        use bounds
        at most exponential in the size of the $(P,X)$,
         and have
        $\mathbf{PSPACE}$-decidable membership
        (given input $(C,P,X)$, decide whether $C\in \mathcal{S}_i(P,X)$)
        and emptiness.

        Then the same is true about the counting sets
        $\mathcal{S}_1(P,X)\cap\mathcal{S}_2(P,X)$,
        $\mathcal{S}_1(P,X)\cup\mathcal{S}_2(P,X)$,
        $\overline{\mathcal{S}_1(P,X)}$,\\
        $\mathrm{pre}^*(\mathcal{S}_1(P,X))$,
        $\mathrm{post}^*(\mathcal{S}_1(P,X))$.
\end{lemma}

\begin{proof}[Proof of the theorem~\ref{thm:io-con-verify}]
        Given the output condition $\psi$
        and a protocol,
        we have a counting set $\hat{\psi}$
        of configurations satisfying $\psi$.
        The protocol ensures $\psi$
        if{}f 
        each reachable configuration
        can still reach a configuration in $\hat{\psi}$.
        In other words,
        no input configuration can reach a configuration
        outside the pre-image of $\hat{\psi}$.
        This can be expressed 
        as emptiness of $\mathcal{I}\cap\mathrm{pre}^*(\overline{\mathrm{pre}^*(\hat{\psi})})$
        where $\mathcal{I}$ is the set of input configurations.
        As $\hat{\psi}$ and $\mathcal{I}$ are decidable in $\mathbf{PSPACE}$,
        repeated application of the lemma 
        yields $\mathbf{PSPACE}$-decidability 
        of the desired emptiness.
        This concludes the proof.
\end{proof}

\begin{remark}
        The proof of $\mathbf{PSPACE}$-hardness
        of verification of immediate observation protocols
        given in \cite{conf/apn/EsparzaRW19}
        uses the constantly false protocol,
        thus it can be interpreted as hardness
        of verifying whether a protocol ensures $D\mapsto{}D(true)=0$.
\end{remark}

\section{Conclusion}

We have introduced a notion of constructive expressive power 
for population protocols and have shown that both for general population
protocols and for immediate observation population protocols
it coincides with the expressive power in the classical setting 
of computing predicates.
We have also shown that the relatively low verification complexity
for immediate observation protocols is preserved.

The aim of being able to verify deployment strategies
suggests 
further work in the direction of modelling failures,
as well as self-stabilisation (i.e. making all states input states).
On the other hand, deployment strategies often operate 
on heterogeneous fleets,
requiring input-output conditions instead of pure output conditions
to verify
(e.g. we want to assign a file-server role to some server with sufficient
storage attached).

Input-output conditions also seem to be a promising direction
for achieving generic composition of population protocol.

\bibliography{construction-population-protocols}

\appendix

\end{document}